\documentclass[a4paper,12pt]{article}
\usepackage[pagewise]{lineno}
\usepackage[T1]{fontenc}
\usepackage{amsmath}
\usepackage{amssymb}
\usepackage{amsfonts}
\usepackage{euscript}
\usepackage{fancyhdr}
\usepackage{graphicx}
\usepackage{color}
\usepackage{watermark}
\usepackage{fancybox}
\usepackage{colortbl}
\usepackage{empheq}
\usepackage[english]{babel}
\usepackage[english]{minitoc}
\usepackage[ruled,vlined,linesnumbered]{algorithm2e}
\usepackage{hyperref}

\RequirePackage[latin1,utf8]{inputenc}

\newcommand{\word}[1]{\mathbf{#1}}
\newcommand{\CC}{\mathcal{C}}

\newcommand{\yv}{\word{y}}
\newcommand{\ev}{\word{e}}

\newcommand{\N}{{{\mathbb N}}}

\makeindex 
\newtheorem{theorem}{Theorem}[section]
\newtheorem{corollary}[theorem]{Corollary}
\newtheorem{definition}[theorem]{Definition}
\newtheorem{example}[theorem]{Example}
\newtheorem{lemma}[theorem]{Lemma}

\newtheorem{proposition}[theorem]{Proposition}
\newtheorem{remark}[theorem]{Remark}
\newenvironment{proof}[1][Proof]{\noindent\textbf{#1.} }{\ \rule{0.5em}{0.5em}}

\begin{document}

\title{On the Rank Decoding Problem Over Finite Principal Ideal Rings}
\author{{\small Herv\'e Tale Kalachi}\thanks{{\footnotesize Department of
Computer Engineering, National Advanced School of Engineering of Yaounde,
University of Yaounde I, Cameroon. E-mail: \texttt{%
herve.tale@univ-yaounde1.cm}}}, {\small Hermann Tchatchiem Kamche}\thanks{%
{\footnotesize Department of Mathematics, Faculty of Science, University of
Yaounde I, Cameroon. E-mail: \texttt{hermann.tchatchiem@gmail.com}}}}
\maketitle

\begin{abstract}
The rank decoding problem has been the subject of much attention in this
last decade. This problem, which is at the base of the security of
public-key cryptosystems based on rank metric codes, is traditionally
studied over finite fields. But the recent generalizations of certain
classes of rank-metric codes from finite fields to finite rings have
naturally created the interest to tackle the rank decoding problem in the
case of finite rings. 
In this paper, we show that solving the rank decoding problem
over finite principal ideal rings is at least as hard as the rank decoding problem over
finite fields. We also show that computing the minimum rank distance for
linear codes over finite principal ideal rings is equivalent to the same
problem for linear codes over finite fields.
Finally, we provide combinatorial type algorithms for solving the rank decoding problem over finite chain
rings together with their average complexities.
\end{abstract}


\section{Introduction}

Rank metric codes are subspaces whose elements can be seen as matrices and
the distance between two elements is the rank of their difference \cite%
{D78,Gab85}. These codes have received a lot of attentions these recent
years, especially for their applications in space time coding \cite%
{Lusina2003maximum}, network coding \cite{Silva2008rank} and cryptography
\cite{GPT91}. One could see, among others: the definition of a new family of
structured codes equipped with the rank metric together with efficient
decoding algorithms \cite{GMRZ13,Aragon2019low}, the generalizations of
several known classes of structured rank metric codes from finite fields to
other poorer structures like Galois rings \cite{RNP20,Renner2020low} or
finite principal ideal rings \cite{Kamche2019rank,KKF21}, each of these
generalizations coming with efficient decoding algorithms for the new
underlined code families. It is important to note that the question of
whether one can decode a given code or not is quite fundamental in coding
theory and code-based cryptography. The answer of this question is generally
obvious when dealing with a code that has a known structure, since each
structured code generally comes with a decoding algorithm. But if the
structure of the code is unknown, this question is well known as the
``problem of decoding a random linear code''.

\paragraph{Decoding Problem for a Random Linear Code.}

The general framework for setting up a code based cryptosystem \cite%
{M78,Bucerzan2017evolution} is to take a generator matrix $\boldsymbol{G}$
of a structured linear code that will undergo some transformations, giving
place to a new generator matrix $\boldsymbol{G}_{pub}$ that is, by
assumption, indistinguishable from a generator matrix of a random linear
code \cite{CFS01}. The matrix $\boldsymbol{G}_{pub}$ is then published by
Alice together with a correction capacity $t$ that depends of the
transformations applied. To send a message $\mathbf{m}$ to Alice, Bob
generates a random error $\mathbf{e}$ of weight $t$ and send the cryptogram $%
\mathbf{y} = \mathbf{m} \boldsymbol{G}_{pub} + \mathbf{e}$. An attacker that
intercepts $\mathbf{y}$ and wants to find $\mathbf{m}$ must then solve the
problem of decoding a ``random'' linear code. Concretely, the problem of
decoding a random linear code is, given a linear code $\mathcal{C}$ (or its
generator matrix), a vector $\mathbf{y}$ of the ambient space and an integer
$t$, to find a word $\mathbf{c}$ in $\mathcal{C}$ such that the distance
from $\mathbf{y}$ to $\mathbf{c}$ is at most $t$. This problem is well known
as being NP-complete in the Hamming metric \cite{BMT78}, and was also shown
recently to be NP-complete for the Lee metric \cite{Weger2020hardness}.
However, the case that interests us in this article is when the rank metric
is used. In that specific case of rank metric codes, a randomized reduction
from the same problem in the Hamming metric was proposed in \cite{GZ16}.

\paragraph{Solving the Rank Decoding Problem.}

From an algorithmic point of view, there are two main techniques for solving
the Rank Decoding Problem. The oldest one is of combinatorial type and was
introduced in \cite{CS96}. This technique can be seen as a generalization to
the rank metric of information set decoding algorithms \cite{P62} in which
one has to look for a set of positions that contents the error support. Note
that in the Hamming metric, the support of an error is the set of non-zero
positions of that error and, given the support or a slightly larger set
containing the support of an error, one can find the associated error in
polynomial time by solving a linear system. The situation is quiet the same
in the rank metric with the difference that each coordinate of a vector is
seen as a vector with coefficients in a base field and, the support of the
vector is then the vector subspace generated by its coefficients. The
combinatorial algorithms thus works by guessing the support of the error and
then solve a system of linear equations to find the coordinates of the error
components in a basis of that support. The complexity of such an algorithm
is then dominated by the inverse of the probability that a vector subspace
chosen randomly is the good one. These algorithms have undergone several
improvements, first in \cite{OJ02} and very recently in \cite{GRS16,AGHT18}
where the authors guess a slightly bigger vector subspace containing the
support of the error.

Besides combinatorial techniques, there are also algebraic techniques for
which the main idea is to translate the notion of rank into an algebraic
setting. The first approach from \cite{LP06} first reduces the rank decoding
problem to the search of minimum rank codewords in an extended linear code.
This approach has been the subject of several recent improvements in \cite%
{BBCGPSTV20,BBBGNRT20}. Another approach based on linearized polynomials was
also proposed in \cite{GRS16}. It should be noted that the algorithms cited
above only apply to codes whose alphabets are finite fields.


\paragraph{Structured Rank Metric Codes Over Finite Rings.}

An important point for setting up a code-based cryptosystem like in \cite%
{M78} is to have a structured family of codes (that is to say code families
with efficient decoding algorithms). This last decade has seen the birth of
several works going in that direction, in particular for rank metric codes
over finite rings. Tchatchiem and Mouaha \cite{Kamche2019rank} first
proposed a generalisation of the well known family of Gabidulin codes to
finite principal ideal rings. This work was followed by \cite%
{Puchinger2021efficient} where the authors provide an iterative decoding
algorithm for Gabidulin codes over Galois rings with provable quadratic
complexity in the code length. Note that the previous algorithm is similar
to the iterative algorithm presented in \cite{Kamche2019rank} for
interleaved Gabidulin codes. Although Gabidulin codes over finite rings have
many applications in network coding and space time coding like in the case
of finite fields, they are not too much attractive in cryptography when
thinking about the story of their use in the Gabidulin, Paramonov and
Tretjakov (GPT) cryptosystem \cite{GPT91}. To put it in a nutshell, due to
the rich algebraic structure of Gabidulin codes, the original GPT
cryptosystem was drastically broken in a series of structural attacks from
Gibson \cite{G95,G96} and Overbeck \cite{O05,O05a,O08}. Even if several
variants where proposed to avoid these attacks \cite%
{Gabidulin2001modified,Gabidulin2008attacks,Rashwan2010smart,Loidreau2010designing,Rashwan2011security}
almost all of them are now found to be vulnerable \cite%
{Otmani2018improved,Horlemann2018extension,Kalachi2022failure}.

In \cite{Renner2020low}, the recent and promising family of Low-Rank
Parity-Check (LRPC) codes \cite{GMRZ13} was also generalized to the rings of
integers modulo a prime power. This work was followed by the paper of
Renner, Neri, and Puchinger \cite{RNP20} that defined LRPC codes over Galois
rings, the paper of Kamwa, Tale, and Fouotsa \cite{Kamwa2021generalization}
that generalized LRPC codes to the ring of integers modulo a positive
integer and finally the work from \cite{Kamche2021low} where the authors
generalize LRPC codes to finite commutative rings. Note that LRPC codes is
known as having a very poorer algebraic structure and, as a consequence,
their use in code-based cryptography closes the door to structural attacks
and in this case, a cryptanalysis must focus on the problem of solving the
rank decoding problem. This is why they are very attractive in code-based
cryptography and consequently, their recent generalizations to finite rings
have naturally highlighted the possibility of doing cryptography using rank
metric codes over finite rings.

\paragraph{The Rank Decoding Problem Over Finite Rings.}

The existence of interesting code families over finite rings is certainly a
determining element for doing McEliece-like cryptography over finite rings,
but it is essential to make sure that the task of a cryptanalyst facing the
rank decoding problem over finite rings will not be facilitated compared to
finite fields. One could also wonder what more one gains by moving from
finite fields to finite rings. Note that some properties of rank metric over
finite fields do not apply to rank metric over finite rings due to zero
divisors (see Example \ref{ZeroDivisor}) and, as a consequence, the main
technique used to translate the notion of rank into an algebraic setting
like in \cite{OJ02}, is not directly applicable in the case of finite
commutative principal ideal rings (see Example \ref{ImpossibleModeling}).
Remark that this modelling way is at the kernel of the recent improvements
of algebraic algorithms for solving the rank decoding problem over finite
fields \cite{BBCGPSTV20,BBBGNRT20}. Thus, the existence of zero divisors
over finite rings could help to avoid some existing attacks over finite
fields. All these elements make the rank decoding problem over finite rings
very attractive for code-based cryptography and give rise to several
questions around this problem. A natural one is its difficulty compared to
the same problem over finite fields. It would also be interesting to provide
practical algorithms for solving this problem as well as their complexities.
Another question which is generally related to the decoding problem is the
calculation of the minimum distance for a rank metric code over finite rings.

\paragraph{Our Contribution.}

In this paper, we use the structure theorem for finite commutative rings
\cite{Mcdonald1974finite} to show that solving the rank decoding problem
over finite principal ideal rings is equivalent to solve the same problem
over finite chain rings. We then use the socle and the injective envelope of
modules over finite chain rings to show that the rank decoding problem over
finite chain rings is at least as hard as the rank decoding problem over
finite fields. We also show that computing the minimum rank distance for
linear codes over finite principal ideal rings is equivalent to the same
problem for linear codes over finite fields as in the case of hamming metric
\cite{Walker1999algebraic,Greferath2004finite,Dougherty2009mds}.
Furthermore, we provide combinatorial type algorithms similar to \cite%
{GRS16,AGHT18} for solving the rank decoding problem over finite chain
rings. To evaluate the average complexity of our algorithms, we use the
shape of modules to give a formula that allows to count the number of
submodules of fixed rank for a finitely generated module over a finite chain
ring. %
%

\paragraph{Organization of the Paper.}

The rest of the paper is organized as follows: in Sections \ref{sec:prelims}
and \ref{sec:finiteChain} we give some notations, definitions, and mathematical results that are useful throughout the paper. In
Section \ref{sec:RDProblem}, we formally define the rank decoding problem
over finite principal ideal rings as well as its dual version and show that
the two problems are equivalent. The combinatorial type algorithms for
solving the rank decoding problem over finite principal ideal rings are then
presented with a complexity analysis in Section \ref{sec:Solving_RDProblem}
and finally, we conclude in Section \ref{sec:conclusion}.

%
%
%

\section{Preliminaries}

\label{sec:prelims} We denote by ${{\mathbb{N}}}$ the set of positive
integers including $0$, and ${{\mathbb{N}}}^*$ the set ${{\mathbb{N}}}$
excluding $0$. Let $m$ and $n$ be two elements of ${{\mathbb{N}}}^*$ and $R$
a finite commutative principal ideal ring, that is to say, a finite ring in which each ideal is generated by one element. 
The set of all $m \times n$
matrices with entries from $R$ will be denoted by $R^{m\times n}$. In
general, we will use bold uppercase letters for matrices and bold lowercase
letters for vectors. 

In \cite{Gab85}, Gabidulin used Galois extensions of finite fields to give vector representations of matrices and thus defined the notion of rank for vectors. Thanks to \cite{Kamche2019rank}, this notion can be extended to finite principal ideal rings by defining Galois extension for finite principal ideal rings. In this section, we recall how to construct such a Galois extension. Note that by the structure theorem for finite commutative rings \cite[Theorem VI.2]{Mcdonald1974finite}, any finite commutative principal ideal ring can be
decomposed as a direct sum of finite commutative local principal ideal
rings, that is to say, finite chain rings. Galois extensions of finite chain
rings can then be used to construct Galois extensions of finite principal
ideal rings.

\subsection{Galois Extensions of Finite Chain Rings}

\label{subsec:FiniteChainSings}A chain ring is a ring whose ideals
are linearly ordered by inclusion \cite{Mcdonald1974finite}. A finite chain
ring have exactly one maximal ideal with is generated by one element. As an
example, given $k \in {{\mathbb{N}}}^*$ and a prime number $p$, the ring $%
\mathbb{Z}_{p^{k}}=\mathbb{Z}/p^{k}\mathbb{Z}$ of integers modulo $p^{k}$ is
a finite chain ring with $p\mathbb{Z}_{p^{k}}$ as the unique maximal ideal.

In this subsection, we assume that $R$ is a finite commutative chain ring
with maximal ideal $\mathfrak{m}$ and residue field $\mathbb{F}_{q}=R/%
\mathfrak{m}$. Let $\pi $ be a generator of $\mathfrak{m}$ and $\nu $ the
nilpotency index of $\pi $, i.e., $\nu$ is the smallest element of ${{%
\mathbb{N}}}^*$ such that $\pi ^{\nu }=0$. Then, any element $a$ in $R$ can
be decomposed into $a=\pi ^{i}u$ where $u$ is a unit of $R$ and $i$ is a
unique element in $\left\{ 0,\ldots ,\nu \right\} $. The natural projection $%
R\rightarrow R/\mathfrak{m} $ is denoted by $\Psi $ and can be extended
coefficient-by-coefficient to polynomials over $R$.

Let $h\in $ $R\left[ X\right] $ be a monic polynomial of degree $m$ such
that $\Psi \left( h\right) $ is irreducible in $\mathbb{F}_{q}\left[ X\right]
$. Set $S=R\left[ X\right] /\left( h\right) $, where $\left( h\right) $
denotes the ideal of $R\left[ X\right]$ generated by $h$. Then, $S$ is a
local Galois extension of $R$ of degree $m$, with maximal ideal $\mathfrak{M}%
=\mathfrak{m}S$ and residue field $\mathbb{F}_{q^{m}}=S/\mathfrak{M}$. Also
note that $S$ can be seen as a free $R-$module of rank $m$. Since $R$ is a
finite chain ring, $S$ is also a finite chain ring and $\pi$ is a generator
of $\mathfrak{M}$. A Galois extension of $\mathbb{Z}_{p^{k}}$ is called a
Galois ring. We refer the reader to \cite{Mcdonald1974finite} for more
details about Galois extensions of finite chain rings, where a
characterization of finite chain rings using Galois rings is also given in
\cite[Theorem XVII.5]{Mcdonald1974finite}. The following example provides a construction of a Galois extension of $\mathbb{Z}_{8}$ of degree $4$. 

\begin{example}
\label{ExampleGaloisExtension} Let $R=\mathbb{Z}_{8}$ and $%
h=X^{4}+4X^{3}+6X^{2}+3X+1\in R\left[ X\right] $. Then $\Psi \left( h\right)
=\allowbreak X^{4}+X+1$ irreducible in $\mathbb{F}_{2}\left[ X\right] $.
Therefore, $S=R\left[ X\right] /\left( h\right) $ is a Galois extension of $%
R$ of degree $4$.
\end{example}

\subsection{Galois Extension of Finite Principal Ideal Rings}

\label{subsec: GaloisExtensionOfFPIR}As previously said, a principal ideal
ring $R$ is isomorphic to a product of finite chain rings. That is to say,
there exists a positive integer $\rho $ such that $R\cong R_{(1)}\times
\cdots \times R_{(\rho )}$, where each $R_{(j)}$ is a finite chain ring.
Using this isomorphism, we identify $R$ with $R_{(1)}\times \cdots \times
R_{(\rho )}$. As an example, if $\eta =p_{1}^{k_{1}}\times \cdots \times
p_{d}^{k_{d}}$ where $p_{1},\ldots ,p_{d}$ are prime numbers and $k_1,
\ldots, k_d$ belonging to ${{\mathbb{N}}}^*$, then the ring $\mathbb{Z}%
_{\eta }$ is isomorphic to the product of finite chain rings $\mathbb{Z}%
_{p_{1}^{k_{1}}},\ldots ,\mathbb{Z}_{p_{d}^{k_{d}}}$, that is to say $%
\mathbb{Z}_{\eta }\cong \mathbb{Z}_{p_{1}^{k_{1}}}\times \cdots \times
\mathbb{Z}_{p_{d}^{k_{d}}}$.

For any $j\in \left\{ 1,\ldots ,\rho \right\} $, since $R_{(j)}$ is a finite
chain rings, let $S_{(j)}$ be a Galois extension of $R_{(j)}$ of degree $m$.
$S:=S_{(1)}\times \cdots \times S_{(\rho )}$ is a free $R-$module of degree $%
m $. Since each $S_{(j)}$ is a finite principal ideal ring, $S$ is also a
finite principal ideal ring. Furthermore, as specified in \cite[pp.7720]%
{Kamche2019rank}, $S$ is a Galois extension of $R$ and there exists a monic
polynomial $h\in R[X]$ of degree $m$ such that $S\cong R[X]/(h)$.

\begin{example}
Let us construct a Galois extension of $R=\mathbb{Z}/40%
\mathbb{Z}$ of degree $4$. Let $R_{(1)}=\mathbb{Z}/5\mathbb{Z}$ and $R_{(2)}=%
\mathbb{Z}/8\mathbb{Z}$. The map $\Phi :R\rightarrow R_{(1)}\times R_{(2)}$
given by $x+40\mathbb{Z}\longmapsto \left( x+5\mathbb{Z},x+8\mathbb{Z}%
\right) $ is a ring isomorphism and its inverse $\Phi ^{-1}$ is defined by 
$\left( x+5 \mathbb{Z},y+8\mathbb{Z}\right) \longmapsto xe_{1}+ye_{2}$, 
where $e_{1}=16+40\mathbb{Z}$ and $e_{2}=25+40\mathbb{Z}$.  Consider $h_{\left( 1\right)
}=X^{4}+4X^{2}+4X+2\in R_{(1)}\left[ X\right] $, $h_{\left( 2\right)
}=X^{4}+4X^{3}+6X^{2}+3X+1\in R_{(2)}\left[ X\right] $, $S_{(1)}=R_{(1)}%
\left[ X\right] /\left( h_{\left( 1\right) }\right) $, and $S_{(2)}=R_{(2)}%
\left[ X\right] /\left( h_{\left( 2\right) }\right) $. Since $R_{(1)}$ is a
finite field and $h_{(1)}$ is irreducible over $R_{(1)}$, then $S_{(1)}$ is
a Galois extension of $R_{(1)}$ of degree $4$. Furthermore, from Example \ref%
{ExampleGaloisExtension}, $S_{(2)}$ is also a Galois extension of $R_{(2)}$ of
degree $4$ and then, $S_{(1)}\times S_{(2)}$ is a Galois extension of $%
R_{(1)}\times R_{(2)}$ of degree $4$. If we extend $\Phi ^{-1}$
coefficient-by-coefficient to $R_{(1)}\left[ X\right] \times $ $R_{(2)}\left[
X\right] $, by taking $h=\Phi ^{-1}\left( h_{\left( 1\right) },h_{\left(
2\right) }\right) =\allowbreak X^{4}+20X^{3}+14X^{2}+19X+17$ we have $%
S_{(1)}\times S_{(2)}\cong R\left[ X\right] /\left( h\right) $ and so, $R%
\left[ X\right] /\left( h\right) $ is a Galois extension of $R$ of degree $4$%
.
\end{example}

\subsection{Rank Metric Codes Over Finite Principal Ideal Rings}
An introduction to rank metric codes over finite principal ideal ring can be found in \cite{Kamche2019rank}. Here we give some fundamental notions needed for the sequel of the paper. Let us start by recalling the following definitions of the rank for a module over a finite principal ideal ring, the rank for a matrix, and a vector with coefficients in
a finite principal ideal ring. Note that the notion of rank for a module is a generalization of the well known notion of dimension for a vector space. So we have the following definition.

\begin{definition}[Rank of a Module]
Let $M$ be a finitely generated $R-$module. The rank of $M$, denoted by $%
rk_{R}\left( M\right) $ or simply $rk\left( M\right) $, is the smallest
number of elements in $M$ generating $M$ as an $R-$module. The rank of the module $%
\left\{ 0\right\} $ is by convention $0$.
\end{definition}

Since the columns (or the rows) of a matrix with coefficients in $R$
generate an $R-$module, the previous notion of rank for a $R-$module
naturally extends to matrices with coefficients in $R$.

\begin{definition}[Rank of a matrix]
Let $\boldsymbol{A}\in R^{m\times n}$. The rank of $\boldsymbol{A}$ ,
denoted by $rk_{R}\left( \boldsymbol{A}\right) $, or simply $rk\left(
\boldsymbol{A}\right) $, is the rank of the $R-$submodule generated by the
column vectors (or row vectors) of $\boldsymbol{A}$.
\end{definition}

A simple way to compute the rank of a matrix from $R^{m\times n}$ is to
compute its Smith normal form and count the number of non-zero elements on
the diagonal (see \cite[Proposition 3.4]{Kamche2019rank}). Also remark that
thanks to the notion of Galois extension for finite principal ideal rings, $%
R^m$ is isomorphic to a Galois extension $S$ of $R$ so that each element of $%
S$ can be considered as a vector of the $R-$module $R^m$. Consequently, we
have the following definition that defines the rank for vectors in $S^n$.

\begin{definition}[Rank of a vector]
Let $\mathbf{u}=\left( u_{1},\ldots ,u_{n}\right) \in S^{n}$.

1) The support of $\mathbf{u}$, denoted $supp(\mathbf{u)}$, is the $R-$%
submodule of $S$ generated by $\left\{ u_{1},\ldots ,u_{n}\right\} $.

2) The rank of $\mathbf{u}$, denoted $rk_{R}\left( \mathbf{u}\right) $, or
simply $rk\left( \mathbf{u}\right) $, is the rank of the support of $\mathbf{%
u}$.
\end{definition}

So, as in the case of fields, the map $S^{n}\times S^{n}\rightarrow\mathbb{N}
$ given by $\left( \mathbf{u,v}\right) \mapsto rk_{R}\left( \mathbf{u-v}%
\right) $ is a metric \cite{Kamche2019rank}. Also note that the rank of a
vector can be computed using its matrix representation. Indeed, since $S$
is also a free $R-$module, let $\left( b_{1},\ldots ,b_{m}\right) $ be a
basis of $S$ and consider $\mathbf{a}=\left( a_{1},\ldots ,a_{n}\right) \in
S^{n}$. For $j=1,\ldots ,n$, $a_{j}$ can be written as $a_{j}=\sum_{1\leq
i\leq m}a_{i,j}b_{i}$, where $a_{i,j}\in R$. The matrix $\mathbf{A=}\left(
a_{i,j}\right) _{1\leq i\leq m,\ 1\leq j\leq n}$ is then the matrix
representation of $\mathbf{a}$ in the $R-$basis $\left( b_{1},\ldots
,b_{m}\right) $ and $rk_{R}\left( \mathbf{a}\right) =rk_{R}\left(
\boldsymbol{A}\right) $.

It is important to underline the fact that some properties of the rank for
matrices over finite fields do not generalize for matrices over finite rings
due to zero divisors. As an example, the rank of a matrix $\boldsymbol{A}$
with entries in a field $F$ is the order of a highest order non-vanishing
minor of $\boldsymbol{A}$ and for any non-zero element $\alpha $ from $F$,
both $\boldsymbol{A}$ and $\alpha \boldsymbol{A}$ have the same rank.
However, those properties are not always true in finite rings.

\begin{example}
\label{ZeroDivisor} Let $\boldsymbol{A}=\left(
\begin{array}{cc}
2 & 0 \\
0 & 2%
\end{array}%
\right) $ be a matrix with entries in $\mathbb{Z}_{4}$. Since $\boldsymbol{A}
$ is in the Smith normal form, $rk(\mathbf{A)}=2$ while $rk(2 \boldsymbol{A}%
)=0$. Moreover, the order of a highest order non-vanishing
minor of $\boldsymbol{A}$ is $1$, which is different from the rank of $\boldsymbol{A}$.
\end{example}

Since $R=R_{(1)}\times \cdots \times R_{(\rho )}$ and $%
S=S_{(1)}\times \cdots \times S_{(\rho )}$, for any $j \in \left\{ 1,\ldots
,\rho \right\} $, we denote by $\Phi _{(j)}$ the $j$-th projection map from $%
S$ to $S_{(j)}$ in the following. We will also extend $\Phi _{(j)}$
coefficient-by-coefficient as a map from $S^{n}$ to $S_{(j)}^{n}$ and, the
restriction of $\Phi _{(j)}$ to $R^{n}$ will also be denoted by $\Phi _{(j)}$%
. We then have the following result from \cite{Dougherty2009mds} :

\begin{proposition}
\label{LocalizationOfRank} For any submodule $N$ of $R^{n}$ he have,%
\begin{equation*}
rk_{R}\left( N\right) =\max_{1\leq j\leq \rho }\left\{ rk_{R_{(j)}}\left(
\Phi _{(j)}\left( N\right) \right) \right\} .
\end{equation*}
\end{proposition}

\begin{proof}
See \cite[Corollary 2.5]{Dougherty2009mds}.
\end{proof}

The above proposition shows that computing the rank of a
submodule $N$ over a finite principal ideal ring is equivalent to compute
the highest rank for the projections of $N$ as submodules over finite chain
rings. This result does apply also to vectors from $S^n$ as they can be
viewed as $R-$submodules when computing their ranks.

\begin{corollary}
\label{LocalizationOfRankVector}For any $\mathbf{a}\in S^{n}$,
\begin{equation*}
rk\left( \mathbf{a}\right) =\max_{1\leq j\leq \rho }\left\{ rk\left( \Phi
_{(j)}\left( \mathbf{a}\right) \right) \right\} .
\end{equation*}
\end{corollary}

Let us recall that an $S$-submodule $\mathcal{C}$ of $S^{n}$ is also called
a linear code of length $n$ over $S$. Its rank will be denoted by $k(%
\mathcal{C)}$ and, its minimum rank distance is $d\left( \mathcal{C}\right)
:=\min \left\{ rk\left( \mathbf{u}-\mathbf{v}\right) :\ \mathbf{u},\mathbf{v}%
\in \mathcal{C},\ \mathbf{u}\neq \mathbf{v}\right\} $. A generator matrix of
$\mathcal{C}$ is any $k(\mathcal{C)}\times n$ matrix over $S$ whose rows
generate $\mathcal{C}$. The dual of $\mathcal{C}$ denoted by $\mathcal{C}%
^{\perp }$ is the orthogonal of $\mathcal{C}$ with respect to the usual
Euclidean inner product on $S^{n}$ and, a parity-check matrix of $\mathcal{C}
$ is a generator matrix of $\mathcal{C}^{\perp }$. By \cite[Proposition 2.9]%
{Fan2014matrix}, if $\mathcal{C}$ is a free module, then $\mathcal{C}^{\perp
}$ is also a free module of rank $n-k(\mathcal{C)}$.

The minimum rank distance $d\left( \mathcal{C}\right) $ is an essential
parameter  for the code $\mathcal{C}$. It allows to evaluate the error correction capacity of $\CC$ which is given by $\lfloor \left( d\left( \mathcal{C}\right) - 1 \right) / 2 \rfloor$. The Singleton bound in rank metric is given by

\begin{equation*}
\log _{\left\vert R\right\vert }\left\vert \mathcal{C}\right\vert \leq \max
\{ m, n \} (\min \{m, n \}-d\left( \mathcal{C}\right) +1).
\end{equation*}%
Codes that achieve this bound are called Maximum Rank Distance (MRD) codes. Note that if $\mathcal{C}$ is a free $S$-submodule of $S^{n}$
then, $\log _{\left\vert R\right\vert }\left\vert \mathcal{C}\right\vert =k(%
\mathcal{C)}m$ holds. Similar to \cite[Lemmas 6.1 and 6.2]{Dougherty2009mds}, we
have the following proposition :

\begin{proposition}
\label{RankDimesionDecomposition}Consider a linear code $\mathcal{C}$ of
length $n$ over $S$ and set $\mathcal{C}_{(j)}:=\Phi _{(j)}\left( \mathcal{C}%
\right) $ for $j=1,\ldots ,\rho $. We have
\begin{equation}
k(\mathcal{C)}=\max_{1\leq j\leq \rho }\left\{ k\left( \mathcal{C}%
_{(j)}\right) \right\}  \label{RankDecoposition}
\end{equation}%
and%
\begin{equation}
d\left( \mathcal{C}\right) =\min_{1\leq j\leq \rho }\left\{ d\left( \mathcal{%
C}_{(j)}\right) \right\} .  \label{DistanceDecomposition}
\end{equation}
\end{proposition}

\begin{proof}
Relation (\ref{RankDecoposition}) is a direct consequence of Proposition \ref%
{LocalizationOfRank}. For relation (\ref{DistanceDecomposition}), let $%
j_{0}\in \left\{ 1,\ldots ,\rho \right\} $ such that $d\left( \mathcal{C}%
_{(j_{0})}\right) =$\ $\min_{1\leq j\leq \rho }\left\{ d\left( \mathcal{C}%
_{(j)}\right) \right\} $, and $\mathbf{c} \in $ $\mathcal{C}$ such that $%
rk\left( \Phi _{(j_{0})}\left( \mathbf{c}\right) \right) =d\left( \mathcal{C}%
_{(j_{0})}\right) $. Consider $\alpha =\left( \alpha _{1},\ldots ,\alpha
_{\rho }\right) \in S$ such that $\alpha _{j_{0}}=1$ and $\alpha _{j}=0$ if $%
j\in \left\{ 1,\ldots ,\rho \right\} \backslash \left\{ j_{0}\right\} $.
Then, $\Phi _{(j_{0})}\left( \alpha \mathbf{c}\right) =\Phi _{(j_{0})}\left(
\mathbf{c}\right) $ and $\Phi _{(j)}\left( \alpha \mathbf{c}\right) =\mathbf{%
0}$ if $j\in \left\{ 1,\ldots ,\rho \right\} \backslash \left\{
j_{0}\right\} $. Therefore, by Corollary \ref{LocalizationOfRankVector}, $%
rk\left( \alpha \mathbf{c}\right) =d\left( \mathcal{C}_{(j_{0})}\right) $.
So, $d\left( \mathcal{C}\right) \leq d\left( \mathcal{C}_{(j_{0})}\right) $.

Let $\mathbf{x}\in \mathcal{C}$ such that $rk\left( \mathbf{x}\right)
=d\left( \mathcal{C}\right) $, then there is $j_{1}\in \left\{ 1,\ldots
,\rho \right\} $ such that $\Phi _{(j_{1})}\left( \mathbf{x}\right) \neq
\mathbf{0}$. Since $rk\left( \mathbf{x}\right) \geq rk\left( \Phi
_{(j_{1})}\left( \mathbf{x}\right) \right) $, we have $d\left( \mathcal{C}%
\right) \geq d\left( \mathcal{C}_{(j_{1})}\right) $ and finally, $d\left(
\mathcal{C}\right) \geq d\left( \mathcal{C}_{(j_{0})}\right) $.
\end{proof}

By Proposition \ref{RankDimesionDecomposition}, the problem of computing the
minimum rank distance of linear codes over finite principal ideal rings is
reduced to the same problem for codes over finite chain rings. In the next
section, we will use the socle and the injective envelope of modules over
finite chain rings to show that this problem reduces to finite fields.

\section{\protect\large Some Properties of Linear Codes Over Finite Chain
Rings}

\label{sec:finiteChain} In this section, we assume as in Subsection \ref%
{sec:prelims} that $R$ is a finite commutative chain ring with residue field
$\mathbb{F}_{q}$ and maximal ideal $\mathfrak{m}$ generated by $\pi$ that
has $\nu $ as its nilpotency index. Remark that $S$ is also a finite chain
ring with residue field $\mathbb{F}_{q^{m}}$. The natural projection $S
\rightarrow \mathbb{F}_{q^{m}}$ is also denoted by $\Psi $ and we extend $%
\Psi $ coefficient-by-coefficient as a map from $S^{n}$ to $\mathbb{F}%
_{q^{m}}^{n}$.

\subsection{\protect\large Socle and Injective Envelope of Modules Over
Finite Chain Rings}

Let $M$ be a finitely generated $R$-module. We recall that the socle of $M$
denoted by $soc_R \left( M\right) $ or simply $soc\left( M\right) $, is the
sum of the minimum nonzero submodules of $M$; while the injective envelope $%
E\left( M\right) $ of $M$ is the smallest injective module containing $M$.
We refer the reader to \cite{Anderson2012rings} for more details about
socles and injective envelopes.

\begin{proposition}
\label{SocleModule}For a finitely generated $R$-module $M$, we have
\begin{equation*}
soc\left( M\right) =soc\left( E\left( M\right) \right) =\pi ^{\nu -1}E\left(
M\right) .
\end{equation*}
\end{proposition}

\begin{proof}
From \cite{Anderson2012rings}, $soc\left( M\right) =soc\left( E\left(
M\right) \right) $. By \cite[Theorem 2.3.]{Honold2000linear}, $E\left(
M\right) $ is a free module and $soc\left( E\left( M\right) \right) =\pi
^{\nu -1}E\left( M\right) $.
\end{proof}

Proposition \ref{SocleModule} provides a relation between the socle and the envelope of a module. Assume for example that $M$ is a rank $k$ submodule of a free $R$-module $V$ of rank $n$. Using the Smith normal form, one can compute a basis $\left\{ b_{1},\ldots ,b_{n}\right\} $ of $V$ and $k$ elements $d_{1},\ldots ,d_{k}$ in $R$ such that $\left\{ d_{1}b_{1},\ldots
,d_{k}b_{k}\right\} $ generates $M$ \cite[Proposition 3.2]{Kamche2019rank}.
Consequently, $E \left( M \right) $ is generated by $\left\{ b_{1},\ldots
,b_{k}\right\} $ and $soc\left( M \right) $ is generated by $\left\{ \pi
^{\nu -1}b_{1},\ldots ,\pi ^{\nu -1}b_{k}\right\} $.

The following proposition continues by showing that any linear code $%
\mathcal{C}$ over $S$ shares the same rank and the same minimum distance
with its socle and its injective envelope.

\begin{proposition}
\label{SocleEnvelepe}For a linear code $\mathcal{C}$ of length $n$ over $S$,
we have
\[ 
k(\mathcal{C})=k\left( soc\left( \mathcal{C}\right) \right) = k \left( E
\left( \mathcal{C}\right) \right) \text{ and } d\left( \mathcal{C}\right)
=d\left( soc\left( \mathcal{C}\right) \right) = d \left( E \left( \mathcal{C}%
\right) \right)
\]
\end{proposition}

\begin{proof}
The proof of the equality $d\left( \mathcal{C}\right) =d\left( soc\left(
\mathcal{C}\right) \right) $ is similar to the proof given in \cite[%
Proposition 3.1]{Greferath2004finite}. Furthermore, since $soc\left(
\mathcal{C}\right) =soc\left( E\left( \mathcal{C}\right) \right) $, we also
have $d \left( E\left( \mathcal{C}\right) \right) =d\left( soc\left( E\left(
\mathcal{C}\right) \right) \right) =d(soc\left( \mathcal{C}\right) )$ and
thanks to \cite[Theorem 2.3]{Honold2000linear}, $k(\mathcal{C})=k(soc\left(
\mathcal{C}\right) )=k(E\left( \mathcal{C}\right) \mathcal{)}$.
\end{proof}

By \cite[Lemma 9]{Kamche2021low} we also have the following :\

\begin{lemma}
\label{Independent}A subset $\left\{ b_{i}\right\} _{1\leq i\leq t}$ of $S$
is $R$-linearly independent if and only if $\left\{ \Psi \left( b_{i}\right)
\right\} _{1\leq i\leq t}$ is $\mathbb{F}_{q}$-linearly independent.
\end{lemma}

Lemma \ref{Independent} states that, showing the $R-$linear independence of
a family of elements in $S$ is equivalent to show the ${{\mathbb{F}}}_q-$%
linear independence of its projection on the residue field. This result is
very useful as it will help to proof several other results starting from the
following lemma.

\begin{lemma}
\label{RankSocleProjection}For any $\mathbf{a}\in S^{n}$, $rk\left( \pi^{\nu
-1}\mathbf{a}\right) =rk\left( \Psi \left( \mathbf{a}\right) \right) .$
\end{lemma}

\begin{proof}
By \cite[Proposition 3.2]{Kamche2019rank}, there exist a basis $\left\{
b_{i}\right\} _{1\leq i\leq m}$ of $S$ and $r=rk\left( \mathbf{a}\right) $
integers $k_{1},k_{2},\ldots ,k_{r} \in {{\mathbb{N}}}$ such that $\{\pi
^{k_{i}}b_{i}\}_{1\leq i\leq r}$ generates $supp(\mathbf{a)}$ with $%
k_{1}\leq k_{2}\leq \cdots \leq k_{r}$. If $k_{1}\neq 0$, then $\pi ^{\nu -1}%
\mathbf{a}=\mathbf{0}$ and $\Psi \left( \mathbf{a}\right) =\mathbf{0}$.
Assume $k_{1}=0$ and let $t$ be the maximum integer in $\left\{ 1,\ldots
,r\right\} $ such that $k_{t}=0$. Then, $\{\pi ^{\nu -1}b_{i}\}_{1\leq i\leq
t}$ is a minimal generating family of $supp(\pi ^{\nu -1}\mathbf{a)}$, that
is to say $rk\left( \pi ^{\nu -1}\mathbf{a}\right) =t$. Moreover $\{\Psi
\left( b_{i}\right) \}_{1\leq i\leq t}$ is a generating family of $%
supp\left( \Psi \left( \mathbf{a}\right) \right) $ and since $\{ b_{i}
\}_{1\leq i\leq t}$ is $R-$linearly independent, thanks to Lemma \ref%
{Independent}, $rk\left( \Psi \left( \mathbf{a}\right) \right) =t$.
\end{proof}

\begin{remark}
\label{isometry} Considering $S$ as an $R-$module, one can remark that the
socle of $S$ is given by $soc_{R}(S)= \pi ^{\nu -1} S$. Furthermore, the map
$\phi :soc_{R}(S)\longrightarrow S/\mathfrak{m}S$ given by $\phi ( \pi ^{\nu
-1}u )= u+\mathfrak{m}S$ is an isomorphism of $R/\mathfrak{m}-$vector spaces
and, extending $\phi $ coefficient-by-coefficient from $\pi ^{\nu -1}S^{n}$
to $\mathbb{F}_{q^{m}}^{n}$ provides an isometry between the normed spaces $%
\left( \pi ^{\nu -1}S^{n},rk_{R}\right) $ and $\left(\mathbb{F}%
_{q^{m}}^{n},rk_{\mathbb{F}_{q}}\right) $ thanks to Lemma \ref%
{RankSocleProjection}.
\end{remark}

The following theorem is a rank metric version of \cite[Theorem 3.4]%
{Walker1999algebraic}.

\begin{theorem}
\label{ProjectionCode} Given a linear code $\mathcal{C}$ of length $n$ over $%
S$ such that $\Psi \left( \mathcal{C}\right) \neq \left\{ \mathbf{0}\right\}
$,

(i) $d\left( \mathcal{C}\right) \leq d\left( \Psi \left( \mathcal{C}\right)
\right).$

(ii) if $\mathcal{C}$ is free, then $d\left( \mathcal{C}\right) =d\left(
\Psi \left( \mathcal{C}\right) \right) $.
\end{theorem}

\begin{proof}
(i) Let $\mathbf{a} \in \mathcal{C}$ such that $rk\left( \Psi \left( \mathbf{%
a} \right) \right) = d \left( \Psi \left( \mathcal{C}\right) \right) $. By
Lemma \ref{RankSocleProjection}, we have $rk\left( \pi ^{\nu -1}\mathbf{a}%
\right) =d\left( \Psi \left( \mathcal{C}\right) \right) $ and since $\pi
^{\nu -1}\mathbf{a}\in \mathcal{C}$, $d\left( \mathcal{C}\right) \leq
d\left( \Psi \left( \mathcal{C}\right) \right) $ holds.

(ii) Assume that $\mathcal{C}$ is free. According to Propositions \ref%
{SocleModule} and \ref{SocleEnvelepe} respectively, $soc\left( \mathcal{C}%
\right) =\pi ^{\nu -1}\mathcal{C}$ and $d\left( \mathcal{C}\right) =d\left(
soc\left( \mathcal{C}\right) \right) $ hold. Hence, there exists $\mathbf{a}%
\in \mathcal{C}$ such that $rk\left( \pi ^{\nu -1}\mathbf{a}\right) =d\left(
\mathcal{C}\right) $ and thanks to Lemma \ref{RankSocleProjection}, $%
rk\left( \Psi \left( \mathbf{a}\right) \right) =d\left( \mathcal{C}\right) $
holds. Therefore, $d\left( \Psi \left( \mathcal{C}\right) \right) \leq
d\left( \mathcal{C}\right) $.
\end{proof}

A direct consequence of Theorem \ref{ProjectionCode} is the following:

\begin{corollary}
\label{ListCode}Let $\mathcal{C}$ be a linear rank metric code of length $n$
over $\mathbb{F}_{q^{m}}$, with rank $k$, minimum rank distance $d$ and
generated by $\mathbf{g}_{1},\ldots ,\mathbf{g}_{k}$. For each $j$ in $\left\{ 1,\ldots ,k\right\} $, let $\mathbf{g}%
_{j}^{\prime }$ in $S^{n}$ such that $\Psi \left( \mathbf{g}_{j}^{\prime
}\right) =\mathbf{g}_{j}$  and $%
\mathcal{C}^{\prime }$ be the linear code generated by $\mathbf{g}_{1}^{\prime
},\ldots ,\mathbf{g}_{k}^{\prime }$. Then, $\mathcal{C}^{\prime }$ is a free
linear rank metric code over $S$ of length $n$, rank $k$, and minimum rank
distance $d$. 
\end{corollary}

\begin{proof}
By Lemma \ref{Independent}, $\mathcal{C}^{\prime }$ is a free code of rank $%
k $. So, by Theorem \ref{ProjectionCode}, the minimum rank distance of $%
\mathcal{C}^{\prime }$ is $d$.
\end{proof}

Corollary \ref{ListCode} shows that the problem of computing the minimum distance for a linear rank metric code over finite chain rings is at least as hard as the same problem for rank metric codes over finite fields. Note that the latter is considered as being NP-hard \cite{GZ16}. Another consequence of Corollary \ref{ListCode} is that one can construct MRD codes over $S$ from MRD codes over $\mathbb{F}_{q^{m}}$. A kind of converse of Corollary \ref{ListCode} is given in the following corollary :

\begin{corollary}
\label{MinimumDistance} Let $\mathcal{C}$ be a linear code of length $n$
over $S$. Then $\Psi \left( E\left( \mathcal{C}\right) \right) $ and $%
\mathcal{C}$ have the same rank and the same minimum rank distance.
\end{corollary}

\begin{proof}
By Proposition \ref{SocleEnvelepe}, $d\left( \mathcal{C}\right) =d\left(
E\left( \mathcal{C}\right) \right) $ and $k\left( \mathcal{C}\right)
=k(E\left( \mathcal{C}\right) \mathcal{)}$ hold. Furthermore, we have $%
d\left( E\left( \mathcal{C}\right) \right) =d\left( \Psi \left( E\left(
\mathcal{C}\right) \right) \right) $ according to Theorem \ref%
{ProjectionCode} and, thanks to Lemma \ref{Independent}, $k\left( E\left(
\mathcal{C}\right) \right) =k\left( \Psi \left( E\left( \mathcal{C}\right)
\right) \right) $ holds and we have the result.
\end{proof}

Thanks to Corollaries \ref{MinimumDistance} and \ref{ListCode}, the problem of computing the
minimum rank distance for linear codes over finite chain rings is equivalent to
the same problem for linear codes over finite fields. Nevertheless, from an algorithmic point of view, this problem over finite rings can be worse in practice since the size of the alphabet is
naturally bigger than its projection (which is the residue field).

\begin{example}
\label{LinearCode}
Consider the rings $R=\mathbb{Z}_{8}$ and $S=R\left[ X\right] /\left( h\right) $ defined in Example \ref%
{ExampleGaloisExtension}. For $a=X+\left( h\right) $ and $\overline{a}=\Psi
\left( a\right) $, $S=R\left[ a\right] $, let $\mathbf{g}%
_{1}=(1,0,6a^{3}+5a^{2}+5,7a^{3}+5a^{2}+a+4)$, $\mathbf{g}%
_{2}=(0,1,5a^{3}+5a^{2}+2a+1,5a^{3}+2a^{2}+4a)$ and $\CC=\left\langle \mathbf{g}_{1},2\mathbf{g}_{2}\right\rangle$, that is to say the linear
code generated by $\mathbf{g}_{1}$ and $2\mathbf{g}_{2}$. 
Then, $\CC$ is of length $4$ and rank $2$. We have $%
soc\left( \CC\right) =\left\langle 4\mathbf{g}_{1},4\mathbf{g}%
_{2}\right\rangle $, $E\left( \CC\right) =\left\langle \mathbf{g}%
_{1},\mathbf{g}_{2}\right\rangle $ and $\Psi \left( E\left( \CC%
\right) \right) =\left\langle \Psi \left( \mathbf{g}_{1}\right) ,\Psi \left(
\mathbf{g}_{2}\right) \right\rangle $, with $\Psi \left( \mathbf{g}%
_{1}\right) =(1,0,\overline{a}^{2}+1,\overline{a}^{3}+\overline{a}^{2}+%
\overline{a})$ and $\Psi \left( \mathbf{g}_{2}\right) =(0,1,\overline{a}^{3}+%
\overline{a}^{2}+1,\overline{a}^{3})$. By \cite{Honold2000linear}, $\mathcal{%
C}$ has $2^{20}$ codewords while $\Psi \left( E\left( \CC\right)
\right) $ has only $2^{8}$ codewords. So, it is algorithmically better to compute the minimum
rank distance of $\CC$ via $\Psi \left( E\left( \CC\right) \right) $. Using SageMath \cite{sagemath2022}, we compute the minimum
distance of $\Psi \left( E\left( \CC\right) \right) $ and get $3$.
Thus, by Corollary \ref{SocleEnvelepe}, the minimum rank distance of $%
\CC$ is $3$.
\end{example}

\subsection{Shapes for Modules Over Finite Chain Rings}

A partition of a positive integer $n$ is a decreasing sequence of positive
integers whose sum is $n$. For $\lambda =\left( \lambda _{1},\lambda
_{2},\ldots ,\lambda _{k},0,\ldots \right) $ we will use the notation $%
\lambda \vdash n$ to say that $\lambda$ is a partition of $n$ and, will only
keep the non-zero components of $\lambda$, that is to say $\lambda =\left(
\lambda _{1},\lambda_{2},\ldots ,\lambda _{k}\right) $. The conjugate of a
partition $\lambda $ is the partition $\lambda ^{\prime }$ defined by $%
\lambda _{i}^{\prime }=\left\vert \left\{ j:\lambda _{j}\geq i\right\}
\right\vert $. By \cite[Theorem 2.2.]{Honold2000linear}, we have the
following proposition :

\begin{proposition}
\label{ShapeOfModule}Let $M$ be a finitely generated $R-$module. Then, there
exists a uniquely determined partition $\lambda =\left( \lambda _{1},\lambda
_{2},\ldots ,\lambda _{r}\right) \vdash \log _{q}\left\vert M\right\vert $,
with $\nu $ $\geq \lambda _{1}$ and $\lambda _{r}\neq 0$ such that%
\begin{equation*}
M \cong R/\mathfrak{m}^{\lambda _{1}}\times \cdots \times R/\mathfrak{m}%
^{\lambda _{r}}.
\end{equation*}%
Moreover, $rk\left( M\right) =\lambda _{1}^{\prime }=r$.
\end{proposition}

\begin{definition}[Shape of a module]
The partition $\lambda $ defined in Proposition \ref{ShapeOfModule} is
called the shape of $M$.
\end{definition}

\begin{example}[Shape of a free module]
\label{ShapeFreeModules}Let $\gamma $ be the shape of a free $R-$module of
rank $n$. Then
\begin{equation*}
\gamma =\underset{n}{(\underbrace{\nu ,\ldots ,\nu }})\ \ \ \text{and\ }\ \
\gamma ^{\prime }=\underset{\nu }{(\underbrace{n,\ldots ,n}}).
\end{equation*}
\end{example}

One can remark that the shape of an $R-$module is very related to its cardinality and its rank. The importance of introducing this notion here also comes from the fact that we use it to give the number of submodules of a given module over a finite chain ring. Recall that the number of subspaces of dimension $k$ in a vector space of dimension $n$ over a finite field with $q$ elements is given by the Gaussian binomial
coefficient: 

\begin{equation*}
\left[
\begin{array}{c}
n \\
k%
\end{array}%
\right] _{q}:=\prod\limits_{i=0}^{k-1}\frac{q^{n}-q^{i}}{q^{k}-q^{i}}.
\end{equation*}%
Note that from \cite{Koetter2008coding}, we have
\begin{equation}
q^{k\left( n-k\right) }\leq \left[
\begin{array}{c}
n \\
k%
\end{array}%
\right] _{q}\leq 4q^{k\left( n-k\right) }.  \label{GaussianBinomialBound}
\end{equation}

When dealing with modules over finite chain rings, it is also possible to
count the number of submodules of fixed rank. Thanks to \cite[Theorem 2.4.]%
{Honold2000linear}, we have the following proposition.

\begin{proposition}
\label{CountingSubmodule} Let $M$ be a finitely generated $R-$module of
shape $\lambda $. Let $\mu $ be a partition satisfying $\mu \leq \lambda $,
that is to say $\mu _{j}\leq \lambda _{j}$ for all $j$. The number of
submodules of $M$ of shape $\mu $ is%
\begin{equation*}
\prod\limits_{i=1}^{\nu }q^{\mu _{i+1}^{\prime }\left( \lambda _{i}^{\prime
}-\mu _{i}^{\prime }\right) }\left[
\begin{array}{c}
\lambda _{i}^{\prime }-\mu _{i+1}^{\prime } \\
\mu _{i}^{\prime }-\mu _{i+1}^{\prime }%
\end{array}%
\right] _{q}.
\end{equation*}
\end{proposition}

Also note that according to Proposition \ref{ShapeOfModule}, the rank of an $%
R- $module of shape $\mu $ is $k$ if and only if $\mu_{1}^{\prime }=k$. So,
we have the following corollary.

\begin{corollary}
\label{RankCounting1}Let $M$ be a finitely generated $R-$module of shape $%
\lambda $. For any $k \in {{\mathbb{N}}}$ such that $k \leq rk\left(
M\right) $, the number of submodules of $M$ of rank $k$ is%
\begin{equation*}
\sum\limits_{\underset{\mu _{j}^{\prime }\leq \lambda _{j}^{\prime }\ for\
all\ j}{0=\mu _{\nu +1}^{\prime }\leq \mu _{\nu }^{\prime }\leq \cdots \leq
\mu _{1}^{\prime }=k}}\prod\limits_{i=1}^{\nu }q^{\mu _{i+1}^{\prime }\left(
\lambda _{i}^{\prime }-\mu _{i}^{\prime }\right) }\left[
\begin{array}{c}
\lambda _{i}^{\prime }-\mu _{i+1}^{\prime } \\
\mu _{i}^{\prime }-\mu _{i+1}^{\prime }%
\end{array}%
\right] _{q}.
\end{equation*}
\end{corollary}

We now end this section by the following proposition
expressing the number of submodules of fixed rank for a given free module over a finite
chain ring and also providing upper and lower bounds.

\begin{proposition}
\label{RankCounting2}Let $F$ be a free $R-$module of rank $n$. The number of
submodules of $F$ of rank $k$ is given by%
\begin{equation*}
\beta \left( q,\nu ,k,n\right) :=\sum\limits_{0=l_{\nu +1}\leq l_{\nu }\leq
\cdots \leq l_{1}=k}\prod\limits_{i=1}^{\nu }q^{l_{i+1}\left( n-l_{i}\right)
}\left[
\begin{array}{c}
n-l_{i+1} \\
l_{i}-l_{i+1}%
\end{array}%
\right] _{q}.
\end{equation*}%
Moreover, if $k\leq n/2$ then
\begin{equation*}
q^{\nu k\left( n-k\right) }\leq \beta \left( q,\nu ,k,n\right) \leq 4^{\nu }%
\binom{k+\nu -1}{\nu -1}q^{\nu k\left( n-k\right) }
\end{equation*}%
where $\binom{k+\nu -1}{\nu -1}$ is a binomial coefficient.
\end{proposition}

\begin{proof}
The number $\beta \left( q,\nu ,k,n\right) $ is obtained using Example \ref%
{ShapeFreeModules} and Corollary \ref{RankCounting1}. \newline
Applying (\ref{GaussianBinomialBound}), we have
\begin{equation*}
q^{\sum_{i=1}^{\nu }l_{i}\left( n-l_{i}\right) }\leq \prod\limits_{i=1}^{\nu
}q^{l_{i+1}\left( n-l_{i}\right) }\left[
\begin{array}{c}
n-l_{i+1} \\
l_{i}-l_{i+1}%
\end{array}%
\right] _{q}\leq 4^{\nu }q^{\sum_{i=1}^{\nu }l_{i}\left( n-l_{i}\right) }
\end{equation*}%
If $k\leq n/2$, then $\sum_{i=1}^{\nu }l_{i}\left( n-l_{i}\right) $ is
maximal when $l_{i}=k$ for $i=1,\ldots ,\nu $.\newline
By \cite[Theorem 2.5.1]{Brualdi1977introductory}, the number of partitions $%
\left( l_{1},l_{2},\ldots ,l_{\nu +1}\right) $ such that $l_{1}=k$ and $%
l_{\nu +1}=0$ is $\binom{k+\nu -1}{\nu -1}$. So, the result follows.
\end{proof}

\section{Rank Decoding Problem}

\label{sec:RDProblem} The Rank Decoding Problem over finite principal ideal
rings is an extension of the well known Rank Decoding Problem from finite
fields to finite principal ideal rings. So the main difference is the change
of the alphabet which of course impacts the metric properties and several other aspects
of the problem. For simplicity, this problem will be sometimes called
\textquotedblleft Rank Decoding Problem\textquotedblright\ without
specification of the alphabet we are working with. Recall that $R$ is a
finite principal ideal ring and $S$ is a Galois extension of $R$ as in
Section \ref{sec:prelims}. We have the following definitions.

\begin{definition}[Rank Decoding Problem $\mathcal{RD}$]
Let $\mathcal{C}$ be an $S$-submodule of $S^{n}$, $\mathbf{y}$ an element of
$S^{n}$ and $t \in {{\mathbb{N}}}^*$. The \emph{Rank Decoding Problem} is to
find $\mathbf{e}$ in $S^{n}$ and $\mathbf{c}$ in $\mathcal{C}$ such that $%
\mathbf{y}=\mathbf{c}+\mathbf{e}$ with $rk(\mathbf{e})\leq t$.
\end{definition}

The dual version of this problem uses parity-check matrices and can be
defined as follows.

\begin{definition}[ Rank Syndrome Decoding Problem $\mathcal{RSD}$]
Let $\mathbf{H}\in S^{l\times n}$, $\mathbf{s}$ an element of $S^{l}$ and $t
\in {{\mathbb{N}}}^*$. The \emph{Rank} \emph{Syndrome Decoding Problem} is
to find $\mathbf{e}$ in $S^{n}$ such that $\mathbf{eH}^{\top }\mathbf{=s}$
with $rk(\mathbf{e})\leq t$.
\end{definition}

As in the case of finite fields, solving the $\mathcal{RD}$ problem is
equivalent to solve the $\mathcal{RSD}$ problem. Applying Proposition \ref%
{LocalizationOfRank}, we have the following:

\begin{proposition}
\label{LocalizationOfRDproblem}Let $\mathcal{C}$ be an $S$-submodule of $%
S^{n}$, $\mathbf{y}$ an element of $S^{n}$ and $t \in {{\mathbb{N}}}^*$.
Then there exist $\mathbf{e}$ in $S^{n}$ and $\mathbf{c}$ in $\mathcal{C}$
such that $\mathbf{y}=\mathbf{c}+\mathbf{e}$ with $rk_{R}(\mathbf{e})\leq t$
if and only if for all $j$ in $\left\{ 1,\ldots ,\rho \right\} $, there
exist $\mathbf{e}_{(j)}$ in $S_{(j)}^{n}$ and $\mathbf{c}_{(j)}$ in $\Phi
_{(j)}\left( \mathcal{C}\right) $ such that $\Phi _{(j)}\left( \mathbf{y}%
\right) =\mathbf{c}_{(j)}+\mathbf{e}_{(j)}$ with $rk_{R_{(j)}}(\mathbf{e}%
_{(j)})\leq t$.
\end{proposition}

By Proposition \ref{LocalizationOfRDproblem}, solving the $\mathcal{RD}$
problem over finite principal ideal rings is equivalent to solve the same
problem over finite chain rings. Furthermore, according to Proposition \ref%
{SocleEnvelepe}, solving the $\mathcal{RD}$ problem over finite chain rings
for a linear code $\mathcal{C}$ reduces to solving the same problem for the
free module $E\left( \mathcal{C}\right) $. So, solving the $\mathcal{RD}$
problem over finite principal ideal rings reduces to solving the same
problem for free modules over finite chain rings. The following proposition
gives a relation between the $\mathcal{RD}$ problem over finite chain rings
and the $\mathcal{RD}$ problem over finite fields.

\begin{proposition}
\label{ReductionFiniteFields}Assume as in Section \ref{sec:finiteChain} that
$R$ is a finite chain ring and $\Psi $ is the natural projection $%
S\rightarrow \mathbb{F}_{q^{m}}$. Let $\mathcal{C}$ be a linear rank metric
code of length $n$ over $\mathbb{F}_{q^{m}}$, with rank $k$, minimum rank
distance $d$ and generated by $\mathbf{g}_{1},\ldots ,\mathbf{g}_{k}$. Let $%
\mathbf{g}^{\prime }_{j}$ in $S^{n}$ such that $\Psi \left( \mathbf{g}%
_{j}^{\prime }\right) =\mathbf{g}_{j}$ for $j$ in $\left\{ 1,\ldots
,k\right\} $. Let $\mathcal{C}^{\prime }$ be a linear code generates by $%
\mathbf{g}_{1}^{\prime },\ldots ,\mathbf{g}_{k}^{\prime }$ and $\mathcal{C}%
^{\prime \prime }=soc\left( \mathcal{C}^{\prime }\right) $. Then,
\begin{enumerate}

\item[(a)] $\mathcal{C}^{\prime \prime }$ is a linear rank metric code over $S$ of
length $n$, rank $k$, and minimum rank distance $d$.

\item[(b)] Let $t \in \N^*$, $\mathbf{y}$ an element of $\mathbb{F}%
_{q^{m}}^{n} $, and $\mathbf{y}^{\prime }$ in $S^{n}$ such that $\Psi \left(
\mathbf{y}^{\prime }\right) =\mathbf{y}$. For $\mathbf{y}^{\prime \prime
}=\pi ^{\nu -1}\mathbf{y}^{\prime }$, the following statements are
equivalent.
\begin{enumerate}

\item[(i)] There exist $\mathbf{e}$ in $\mathbb{F}_{q^{m}}^{n}$ and $\mathbf{c}$
in $\mathcal{C}$ such that $\mathbf{y}=\mathbf{c}+\mathbf{e}$ with $rk\left(
\mathbf{e}\right) \leq t$.

\item[(ii)] There exist $\mathbf{e}^{\prime \prime }$ in $S^{n}$ and $\mathbf{c}%
^{\prime \prime }$ in $\mathcal{C}^{\prime \prime }$ such that $\mathbf{y}%
^{\prime \prime }=\mathbf{c}^{\prime \prime }+\mathbf{e}^{\prime \prime }$
with $rk(\mathbf{e}^{\prime \prime })\leq t$.
\end{enumerate}
\end{enumerate}

\end{proposition}

\begin{proof}
(a) By Corollary \ref{ListCode}, $\mathcal{C}^{\prime }$ is a free code of
rank $k$ and minimum rank distance $d$. Thus, thanks to Proposition \ref%
{SocleEnvelepe}, the result follows.

(b) This result is a direct consequence of Remark \ref{isometry}.
\end{proof}

According to Proposition \ref{ReductionFiniteFields}, the $\mathcal{RD}$
problem for a linear code $\mathcal{C}$ over the finite field $\mathbb{F}%
_{q^{m}}$ reduces to solving the same problem for the linear code $\mathcal{C%
}^{\prime \prime }$ over the finite chain ring $S$. This reduction shows
that the $\mathcal{RD}$ problem over finite chain rings is at least as
hard as its finite fields version.

Over finite fields, given an instance $\left( \CC , \yv \right) $ of the $%
\mathcal{RD}$ problem, if the rank of the error is less than the error
correction capability of the linear code $\CC$, then it
is always possible to reduce the $\mathcal{RD}$ problem to the search of
minimum rank codewords in the linear code generated by $\CC \cup \{ \yv \} $, see \cite{OJ02}. This technique is at the base of several
methods for solving the $\mathcal{RD}$ problem over finite fields \cite%
{OJ02,GRS16,AGHT18,BBCGPSTV20,BBBGNRT20}. When dealing with finite rings,
this reduction is generally impossible due to zero divisors. As an
illustration, consider the following example.

\begin{example}
\label{ImpossibleModeling}For 
$
R=\mathbb{Z}_{4},\ \ S=R\left[ X\right] /\left( X^{5}+X^{2}+1\right)$  and 
$a=X+\left( X^{5}+X^{2}+1\right)$, $S$ is a Galois extension of $R$. Let $\mathcal{C}$ be the Gabidulin
code generated by $\mathbf{g}=(1,a,a^{2},a^{3},a^{4})$. By \cite[Theorem 3.24%
]{Kamche2019rank}, the error correction capability of $\mathcal{C}$ is $2$.
Set $\mathbf{e}=(1,2a,0,0,0)$. By \cite{Kamche2019rank}, $rk(\mathbf{e})=2$ and, 
considering the received word $\mathbf{y=e}$, let $\mathcal{C}_{\mathbf{y}}$
be the linear code generated by $\mathbf{g}$ and $\mathbf{y}$. Then $2\mathbf{e%
}=(2,0,0,0,0)\in \mathcal{C}_{\mathbf{y}}$ and $rk(2\mathbf{e})=1$. So, a
solution to the shortest vector problem in the extended code $\mathcal{C}_{%
\mathbf{y}}$ is not a solution to the associated $\mathcal{RD}$ problem as
in \cite{OJ02}.
\end{example}

\section{Solving the Rank Syndrome Decoding Problem}

\label{sec:Solving_RDProblem} According to Proposition \ref%
{LocalizationOfRDproblem} and Proposition \ref{SocleEnvelepe}, we will
restrict the study of the $\mathcal{RD}$ problem to free modules over finite
chain rings. So in what follows, we assume without loss of generality that $\rho =1$. That is to
say, $R$ is a finite chain ring with residue field $\mathbb{F}_{q}$ and $\nu
$ the nilpotency index of its maximal ideal. By \cite[Proposition 3.2]%
{Kamche2019rank}, we have the following lemma:

\begin{lemma}
\label{ExistenceOfF}Let $V$ be a free $R-$module of rank $a$, and $W$ a
submodule of $V$ of rank $b$. For any integer $u$ such that $b\leq u\leq a$,
there exists a free submodule $F$ of $V$ with rank $u$ such that $W\subset F$%
.
\end{lemma}

Lemma \ref{ExistenceOfF} allows to extend the works of \cite{GRS16,AGHT18}
to finite principal ideal rings. Indeed, let $\left( \mathbf{H,s}\right) $
be an instance of the $\mathcal{RSD}$ problem where \ $\mathbf{eH}^{\top }%
\mathbf{=s}$. Let $\mathbf{E}$ be the matrix representation of $\mathbf{e}$
in an $R-$basis of $S$. To recover $\mathbf{e}$, we have two possibilities.
The first possibility is to choose a free $R-$submodule $F$ of $S$ such that
\ $supp(\mathbf{e)}\subset F$. This approach is generally used when $n\geq m$%
. The second possibility is to choose a free $R-$submodule $F$ of $R^{n}$
such that $row(\mathbf{E)}\subset F$, where $row(\mathbf{E)}$ is the $R-$%
submodule generated by the row vectors of $\mathbf{E}$. This approach is
generally used when $m\geq n$. In the following, we give more details on
these combinatorial approaches.

\subsection{First Approach}

We recall that this approach is generally used when $n\geq m$.

\begin{lemma}
\label{InfoSetDeco1} Let $\mathbf{H=}\left( h_{i,j}\right) \in
S^{(n-k)\times n}\ $whose row vectors are linearly independent, $\mathbf{s}$
an element of $S^{n-k}$. Suppose we want to solve an instance $\left(
\mathbf{H,s}\right) $ of the $\mathcal{RSD}$ problem with
\begin{equation}
\mathbf{eH}^{\top }\mathbf{=s}  \label{SyndromeEquation1}
\end{equation}%
where $\mathbf{e=}\left( e_{1},\ldots ,e_{n}\right) \in S^{n}$ and $rk(%
\mathbf{e})=r$. Let $F$ be a free $R-$submodule of $S$ of rank $u$. Assume
that $supp(\mathbf{e)}\subset F$. Let $\left\{ f_{1},\ldots ,f_{u}\right\} $
be a basis of $F$ and $x_{i,j}\in R$ such that, for all $j\in \left\{
1,\ldots ,n\right\} $,
\begin{equation}
e_{j}=\sum\limits_{i=1}^{u}x_{i,j}f_{i}.  \label{ErrorDecomposition1}
\end{equation}%
Then, Equation (\ref{SyndromeEquation1}) with unknown $\mathbf{e}$ can be
transformed into a system of linear equations over $R$ (that we denote by $%
\left( \mathcal{E}_{1}\right) $) with $m\left( n-k\right) $ equations and $%
n\times u$ unknowns $x_{i,j}$.
\end{lemma}

\begin{proof}
Set $\mathbf{X}=\left( x_{i,j}\right) _{1\leq i\leq u,1\leq j\leq n}$ and $%
\mathbf{f=}\left( f_{1},\ldots ,f_{u}\right) $. Then, by (\ref%
{ErrorDecomposition1}), we have
\begin{equation*}
\mathbf{e=fX}\text{.}
\end{equation*}

So, (\ref{SyndromeEquation1}) becomes%
\begin{equation*}
\mathbf{fXH}^{\top }\mathbf{=s}\text{.}
\end{equation*}

Therefore, applying \cite[Lemma 4.3.1]{Horn1991topics}, we have%
\begin{equation}
\left( \mathbf{H}\otimes \mathbf{f}\right) vec\left( \mathbf{X}\right)
=vec\left( \mathbf{s}\right) .  \label{KroneckerProduct1}
\end{equation}

where $\otimes $ is the Kronecker product and $vec\left( \mathbf{X}\right) $
denotes the vectorization of the matrix $\mathbf{X}$, that is to say the
matrix formed by stacking the columns of $\mathbf{X}$ into a single column
vector. Since $S$ is a free $R-$module of rank $m$, (\ref{KroneckerProduct1}%
) can be expanded over $R$ into a linear system with $m\left( n-k\right) $
equations and $n \times u$ unknowns $x_{i,j}$.
\end{proof}

\begin{remark}
\label{Boun1}Let $\boldsymbol{A}$ be the $m\left( n-k\right) \times n u$
matrix which defines Equation $\left( \mathcal{E}_{1}\right) $ of Lemma \ref%
{InfoSetDeco1}.

1) If the column vectors of $\boldsymbol{A}$ are linearly independent, then $%
\left( \mathcal{E}_{1}\right) $ has at most one solution. By \cite[Lemma 2.6
]{Fan2014matrix}, if the column vectors of $\boldsymbol{A}$ are linearly
independent, then $nu\leq m\left( n-k\right) $, that is to say $u\leq
m\left( n-k\right) /n$. So, in practice, we choose $u=\left\lfloor m\left(
n-k\right) /n\right\rfloor $.

2) Assume that $nu\leq m\left( n-k\right) $. Then, by Proposition \ref%
{CountingSubmodule}, the probability that the column vectors of a random $%
m\left( n-k\right) \times nu$ matrix with entries from $R$ are linearly
independent is $\prod_{i=0}^{un-1}\left( 1-q^{i-m\left( n-k\right) }\right) $%
. So, in practice, the column vectors of $\boldsymbol{A}$ are linearly
independent with high probability.
\end{remark}

Lemma \ref{InfoSetDeco1} allows to give Algorithm \ref{DecodingAlgorithm1}.

\begin{algorithm}[h]
\label{DecodingAlgorithm1}
\caption{First Syndrome Decoding Algorithm}
\DontPrintSemicolon
\KwIn{

$\bullet $ $r$   the rank of the error;

$\bullet $ $\mathbf{H\in }S^{(n-k)\times n}\ $whose row vectors are
linearly independent;

$\bullet $ $\mathbf{s}$ an element of $S^{n-k}$ such that there is $\mathbf{%
e\in }S^{n}$ with $rk(\ev)=r \leq u $ and $\mathbf{eH}^{T}\mathbf{=s}$,

where $u:=\left\lfloor m\left( n-k\right) /n\right\rfloor $.

}

\KwOut{an element $\mathbf{e\in }S^{n}$ such that $rk(\ev)=r$ and $%
\mathbf{eH}^{T}\mathbf{=s}$.
}
$update \gets false $ \;

\While{update=false}{

    Choose a free $R-$submodule $F$ of $S$ of rank $u$.

    Choose a basis $\left\{ f_{1},\ldots ,f_{u}\right\} $ of $F$.

    Solve Equation $\left( \mathcal{E}_{1}\right) $ of Lemma \ref{InfoSetDeco1}.

    \eIf{$\left( \mathcal{E}_{1}\right) $ has no solution}{
        $update \gets false $}{
    Use a solution of $\left( \mathcal{E}_{1}\right) $ to compute $\ev$
as in (\ref{ErrorDecomposition1}).

    \eIf{$rk(\ev)\neq r$}{
        $update \gets false $}{
    $update \gets true $}
}
}
\Return{$\ev$}
\end{algorithm}

\begin{theorem}
\label{Complexity1}An average complexity of Algorithm \ref%
{DecodingAlgorithm1} is%
\begin{equation*}
\mathcal{O}\left( m\left( n-k\right) u^{2}n^{2}\beta \left( q,\nu
,r,m\right) /\beta \left( q,\nu ,r,u\right) \right)
\end{equation*}%
operations in $R$, where $\beta \left( q,\nu ,r,m\right) $ and $\beta \left(
q,\nu ,r,u\right) $ are defined in Proposition \ref{RankCounting2}.
\end{theorem}

\begin{proof}
Since $R$ is a finite chain ring, we can use \cite[Algorithm 4.2]%
{Bulyovszky2017polynomial} to solve $\left( \mathcal{E}_{1}\right) $. As $%
un\leq m\left( n-k\right) $, by \cite{Bulyovszky2017polynomial}, $\left(
\mathcal{E}_{1}\right) $ can be solved in $\mathcal{O}\left( m\left(
n-k\right) n^{2}u^{2}\right) $ operations in $R$. So, an average complexity
to recover $\mathbf{e}$ is $\mathcal{O}\left( m\left( n-k\right)
u^{2}n^{2}/p\right) $ where $p$ is the probability that $supp(\mathbf{e)}%
\subset F$. Remark that $p$ is equal to the number of submodules of $S$ of
rank $r$ in a free submodule of $S$ of rank $u$ divided by the number of
submodules of $S$ of rank $r$. By Proposition \ref{RankCounting2}, $p=\beta
\left( q,\nu ,r,u\right) /\beta \left( q,\nu ,r,m\right) $. Thus, the result
follows. 
\end{proof}

\begin{remark}
\label{ApproximationProbability}In practice, we have $r\leq u/2$. Thus, from
Proposition \ref{RankCounting2}, we have
\begin{equation*}
\beta \left( q,\nu ,r,m\right) /\beta \left( q,\nu ,r,u\right) \approx
q^{\nu r\left( m-r\right) }/q^{\nu r\left( u-r\right) }=|R|^{r\left(
m-u\right) } = |R|^{r\left\lfloor mk/n\right\rfloor }
\end{equation*}%
where $\left\vert R\right\vert =q^{\nu }$ is the cardinality of $R$. This
approximation is analogous to the one given in \cite{GRS16} when $R$ is a
finite field with $q$ elements.
\end{remark}

\subsection{Second Approach}

We recall that this approach is generally used when $m\geq n$.

\begin{lemma}
\label{InfoSetDeco2} Let $\mathbf{H=}\left( h_{i,j}\right) \in
S^{(n-k)\times n}$ whose row vectors are linearly independent, $\mathbf{s}$
an element of $S^{n-k}$. Suppose we want to solve an instance $\left(
\mathbf{H,s}\right) $ of the $\mathcal{RSD}$ problem with
\begin{equation}
\mathbf{eH}^{\top }\mathbf{=s}  \label{SyndromeEquation2}
\end{equation}%
where $\mathbf{e}\in S^{n}$ and $rk(\mathbf{e})=r$. Let $\left( b_{1},\ldots
,b_{m}\right) $ be a basis of $S$ as an $R-$module and $\mathbf{E}$ a matrix
representation of $\mathbf{e}$ in this basis. Let $F$ be a free $R-$%
submodule of $R^{n}$ of rank $u$. Assume that $row(\mathbf{E})\subset F$.
Let $\mathbf{F}$ be the $u\times n$ matrix whose row vectors generate $F$
and $\mathbf{X}=\left( x_{i,j}\right) \in R^{m\times u}$ such that
\begin{equation}
\mathbf{E}=\mathbf{XF}  \label{ErrorDecomposition2}
\end{equation}%
Then, Equation (\ref{SyndromeEquation2}) with unknown $\mathbf{e}$ can be
transformed into a system of linear equations over $R$ (that we denote by $%
\left( \mathcal{E}_{2}\right) $) with $m\left( n-k\right) $ equations and $%
mu $ unknowns $x_{i,j}$.
\end{lemma}

\begin{proof}
Set $\mathbf{b}=\left( b_{1},\ldots ,b_{m}\right) $. We have
\begin{equation}
\mathbf{e=bE}.  \label{ErrorDecomposition3}
\end{equation}%
So, (\ref{SyndromeEquation2}) becomes
\begin{equation*}
\mathbf{bXFH}^{\top }\mathbf{=s.}
\end{equation*}%
Therefore, applying \cite[Lemma 4.3.1]{Horn1991topics}, we have

\begin{equation}
\left( \mathbf{HF}^{\top }\otimes \mathbf{b}\right) vec\left( \mathbf{X}%
\right) =vec\left( \mathbf{s}\right) .  \label{KroneckerProduct2}
\end{equation}%
Since $S$ is a free $R-$module of rank $m$, (\ref{KroneckerProduct2}) can be
expanded over $R$ into a linear system with $m\left( n-k\right) $ equations
and $mu$ unknowns $x_{i,j}$.
\end{proof}

\begin{remark}
As in Remark \ref{Boun1}, if the column vectors of the $m\left( n-k\right)
\times um$ matrix which defines Equation $\left( \mathcal{E}_{2}\right) $ of
Lemma \ref{InfoSetDeco2} are linearly independent, then $mu\leq m\left(
n-k\right) $. So, in practice, we choose $u=n-k$.
\end{remark}

Lemma \ref{InfoSetDeco2} allows to give Algorithm \ref{DecodingAlgorithm2}.

\begin{algorithm}[h]
\label{DecodingAlgorithm2}
\caption{Second Syndrome Decoding Algorithm}
\DontPrintSemicolon
\KwIn{

$\bullet $ $r$   the rank of the error;

$\bullet $ $\mathbf{H\in }S^{(n-k)\times n}\ $whose row vectors are
linearly independent;

$\bullet $ $\mathbf{s}$ an element of $S^{n-k}$ such that there is $\mathbf{%
e\in }S^{n}$ with $rk(\ev)=r \leq n-k $ and $\mathbf{eH}^{T}\mathbf{=s}$.

}

\KwOut{an element $\mathbf{e\in }S^{n}$ such that $rk(\ev)=r$ and $%
\mathbf{eH}^{T}\mathbf{=s}$.
}

Choose a basis $\left( b_{1},\ldots ,b_{m}\right) $ of \ $S$ as $R-$module.

$update \gets false $ \;

\While{update=false}{

    Choose a free $R-$submodule $F$ of $R^{n}$ of rank $n-k$.

    Choose a basis $\left\{ \mathbf{F}_{1},\ldots ,\mathbf{F}_{n-k}\right\}$ of $F$.

    Solve Equation $\left( \mathcal{E}_{2}\right) $ of Lemma \ref{InfoSetDeco2}.

    \eIf{$\left( \mathcal{E}_{2}\right) $ has no solution}{
        $update \gets false $}{
    Use a solution of $\left( \mathcal{E}_{2}\right) $ to compute $\ev$
as in (\ref{ErrorDecomposition2}) and (\ref{ErrorDecomposition3}).

    \eIf{$rk(\ev)\neq r$}{
        $update \gets false $}{
    $update \gets true $}
}
}
\Return{$\ev$}
\end{algorithm}

\begin{theorem}
\label{Complexity2}An average complexity of Algorithm \ref%
{DecodingAlgorithm2} is%
\begin{equation*}
\mathcal{O}\left( m^{3}\left( n-k\right)^{3}\beta \left( q,\nu ,r,n\right)
/\beta \left( q,\nu ,r,n-k\right) \right)
\end{equation*}%
operations in $R$, where $\beta \left( q,\nu ,r,n\right) $ and $\beta \left(
q,\nu ,r,n-k\right) $ are defined in Proposition \ref{RankCounting2}.
\end{theorem}

\begin{proof}
The proof is similar to that of Theorem \ref{Complexity1}.
\end{proof}

\begin{remark}
As in Remark \ref{ApproximationProbability}, we have
\begin{equation*}
\beta \left( q,\nu ,r,n\right) /\beta \left( q,\nu ,r,n-k\right) \approx
|R|^{rk}\text{.}
\end{equation*}
\end{remark}

\begin{example}
Consider the linear code $\mathcal{C}$ defined in Example \ref{LinearCode}.
Since the minimum rank distance of $\mathcal{C}$ is $3$, then the error
correction capability of $\mathcal{C}$ is $1$. We consider the received word
\begin{equation*}
\mathbf{y=}\left( 4a^{3}+a^{2}+2a+3,4a^{3}+4,7a+2,6a^{3}+4a^{2}+a+5\right)
\end{equation*}%
of $\mathcal{C}$. Note that $\mathcal{C}$ is not a free module. Thus, to
decode $\mathbf{y}$ we consider that $\mathbf{y}$ is a received word from $%
E\left( \mathcal{C}\right) $. A parity-check matrix of $E\left( \mathcal{C}%
\right) $ is
\begin{equation*}
\mathbf{H}=\left(
\begin{array}{cccc}
6a^{3}+5a^{2}+5 & 5a^{3}+5a^{2}+2a+1 & 7 & 0 \\
7a^{3}+5a^{2}+a+4 & 5a^{3}+2a^{2}+4a & 0 & 7%
\end{array}%
\allowbreak \right)
\end{equation*}%
and the the syndrome of $\mathbf{y}$ is
\begin{equation*}
\mathbf{s=yH}^{\top }=\left( 4a^{3}+4a^{2}+2a,6a^{2}+6a+4\right) .
\end{equation*}%
We run Algorithm \ref{DecodingAlgorithm2} in SageMath \cite{sagemath2022}
with inputs $\mathbf{H}$, $\mathbf{s}$, and $r=1$. This algorithm returns%
\begin{equation*}
\mathbf{e}=(2+6a^{2},0,4+4a^{2},6+2a^{2}).
\end{equation*}%
So the transmitted codeword is
\begin{equation*}
\mathbf{y}-\mathbf{e}=\left(
4a^{3}+3a^{2}+2a+1,4a^{3}+4,4a^{2}+7a+6,6a^{3}+2a^{2}+a+7\right) .
\end{equation*}
\end{example}

\begin{remark}
Theorems \ref{Complexity1} and \ref{Complexity2} give an average complexities
for solving the $\mathcal{RD}$ problem over finite chain rings. According to
Proposition \ref{LocalizationOfRDproblem}, solving the $\mathcal{RD}$
problem over the finite principal ideal ring $R=R_{(1)}\times \cdots \times
R_{(\rho )}$ is equivalent to solving the same problem over each finite
chain ring $R_{(j)}$ for $j$ in $\left\{ 1,\ldots ,\rho \right\} $. So, an
average complexity for solving the $\mathcal{RD}$ problem over the finite
principal ideal ring $R$ is the sum of average complexities over $R_{(j)}$
for $j$ in $\left\{ 1,\ldots ,\rho \right\} $.
\end{remark}

\section{Conclusion}

\label{sec:conclusion} 
We have shown that solving the rank decoding problem
over finite principal ideal rings is at least as hard as the rank decoding problem over
finite fields. We have also shown that computing the minimum rank distance for
linear codes over finite principal ideal rings is equivalent to the same
problem for linear codes over finite fields as in the case of hamming metric
\cite{Walker1999algebraic,Greferath2004finite,Dougherty2009mds}. All these put together with the fact that recent powerful algebraic methods \cite{BBCGPSTV20,BBBGNRT20} for solving the Rank Decoding Problem over finite
fields do not apply directly to finite rings with zero divisors as we have observed in this paper make the $\mathcal{RD}$ problem over finite rings very  promising for code-based cryptography.   

We have also provided combinatorial type algorithms similar to \cite%
{GRS16,AGHT18} for solving the rank decoding problem over finite chain
rings. The average complexities of the underlined algorithms are also given.

An interesting perspective will be to study the cases in which algebraic
algorithms do apply. As an example, one can investigate the possibility of using the properties of linearized polynomials generalized in \cite{Kamche2019rank} to give an algebraic method as in \cite{GRS16} for
solving the rank decoding problem over finite principal ideal rings.

\bibliographystyle{IEEEtran}
\bibliography{codecrypto4}

\begin{thebibliography}{10}
\providecommand{\url}[1]{#1}
\csname url@samestyle\endcsname
\providecommand{\newblock}{\relax}
\providecommand{\bibinfo}[2]{#2}
\providecommand{\BIBentrySTDinterwordspacing}{\spaceskip=0pt\relax}
\providecommand{\BIBentryALTinterwordstretchfactor}{4}
\providecommand{\BIBentryALTinterwordspacing}{\spaceskip=\fontdimen2\font plus
\BIBentryALTinterwordstretchfactor\fontdimen3\font minus
  \fontdimen4\font\relax}
\providecommand{\BIBforeignlanguage}[2]{{%
\expandafter\ifx\csname l@#1\endcsname\relax
\typeout{** WARNING: IEEEtran.bst: No hyphenation pattern has been}%
\typeout{** loaded for the language `#1'. Using the pattern for}%
\typeout{** the default language instead.}%
\else
\language=\csname l@#1\endcsname
\fi
#2}}
\providecommand{\BIBdecl}{\relax}
\BIBdecl

\bibitem{D78}
P.~Delsarte, ``Bilinear forms over a finite field, with applications to coding
  theory,'' \emph{J. Comb. Theory, Ser. {A}}, vol.~25, no.~3, pp. 226--241,
  1978.

\bibitem{Gab85}
{\`E}.~M. Gabidulin, ``Theory of codes with maximum rank distance,''
  \emph{Problemy Peredachi Informatsii}, vol.~21, no.~1, pp. 3--16, 1985.

\bibitem{Lusina2003maximum}
P.~Lusina, E.~Gabidulin, and M.~Bossert, ``Maximum rank distance codes as
  space-time codes,'' \emph{IEEE Transactions on Information Theory}, vol.~49,
  no.~10, pp. 2757--2760, 2003.

\bibitem{Silva2008rank}
D.~Silva, F.~R. Kschischang, and R.~Koetter, ``A rank-metric approach to error
  control in random network coding,'' \emph{IEEE transactions on information
  theory}, vol.~54, no.~9, pp. 3951--3967, 2008.

\bibitem{GPT91}
E.~M. Gabidulin, A.~V. Paramonov, and O.~V. Tretjakov, ``Ideals over a
  non-commutative ring and their applications to cryptography,'' in
  \emph{Advances in Cryptology - EUROCRYPT'91}, ser. Lecture Notes in Comput.
  Sci., no. 547, Brighton, Apr. 1991, pp. 482--489.

\bibitem{GMRZ13}
P.~Gaborit, G.~Murat, O.~Ruatta, and G.~Z{\'e}mor, ``Low rank parity check
  codes and their application to cryptography,'' in \emph{Proceedings of the
  Workshop on Coding and Cryptography WCC'2013}, Bergen, Norway, 2013,
  available on {\tt www.selmer.uib.no/WCC2013/pdfs/Gaborit.pdf}.

\bibitem{Aragon2019low}
N.~Aragon, P.~Gaborit, A.~Hauteville, O.~Ruatta, and G.~Z{\'e}mor, ``Low rank
  parity check codes: New decoding algorithms and applications to
  cryptography,'' \emph{IEEE Transactions on Information Theory}, vol.~65,
  no.~12, pp. 7697--7717, 2019.

\bibitem{RNP20}
J.~Renner, A.~Neri, and S.~Puchinger, ``Low-rank parity-check codes over galois
  rings,'' \emph{Designs, Codes and Cryptography}, pp. 1--36, 2020.

\bibitem{Renner2020low}
J.~Renner, S.~Puchinger, A.~Wachter{-}Zeh, C.~Hollanti, and
  R.~Freij{-}Hollanti, ``Low-rank parity-check codes over the ring of integers
  modulo a prime power,'' in \emph{{IEEE} International Symposium on
  Information Theory, {ISIT} 2020, Los Angeles, CA, USA, June 21-26,
  2020}.\hskip 1em plus 0.5em minus 0.4em\relax {IEEE}, 2020, pp. 19--24.

\bibitem{Kamche2019rank}
H.~T. Kamche and C.~Mouaha, ``Rank-metric codes over finite principal ideal
  rings and applications,'' \emph{IEEE Transactions on Information Theory},
  vol.~65, no.~12, pp. 7718--7735, 2019.

\bibitem{KKF21}
H.~Bartz, L.~Holzbaur, H.~Liu, S.~Puchinger, J.~Renner, A.~Wachter-Zeh
  \emph{et~al.}, ``Rank-metric codes and their applications,''
  \emph{Foundations and Trends in Communications and Information Theory},
  vol.~19, no.~3, pp. 390--546, 2022.

\bibitem{M78}
R.~J. McEliece, \emph{A Public-Key System Based on Algebraic Coding
  Theory}.\hskip 1em plus 0.5em minus 0.4em\relax Jet Propulsion Lab, 1978, pp.
  114--116, dSN Progress Report 44.

\bibitem{Bucerzan2017evolution}
D.~Bucerzan, V.~Dragoi, and H.~T. Kalachi, ``Evolution of the mceliece public
  key encryption scheme,'' in \emph{International Conference for Information
  Technology and Communications}.\hskip 1em plus 0.5em minus 0.4em\relax
  Springer, 2017, pp. 129--149.

\bibitem{CFS01}
N.~Courtois, M.~Finiasz, and N.~Sendrier, ``How to achieve a {McEliece}-based
  digital signature scheme,'' in \emph{Advances in Cryptology -
  ASIACRYPT~2001}, ser. Lecture Notes in Comput. Sci., vol. 2248.\hskip 1em
  plus 0.5em minus 0.4em\relax Gold Coast, Australia: Springer, 2001, pp.
  157--174.

\bibitem{BMT78}
E.~Berlekamp, R.~McEliece, and H.~van Tilborg, ``On the inherent intractability
  of certain coding problems,'' \emph{IEEE Trans. Inform. Theory}, vol.~24,
  no.~3, pp. 384--386, May 1978.

\bibitem{Weger2020hardness}
V.~Weger, K.~Khathuria, A.-L. Horlemann, M.~Battaglioni, P.~Santini, and
  E.~Persichetti, ``On the hardness of the lee syndrome decoding problem,''
  \emph{arXiv preprint arXiv:2002.12785}, 2020.

\bibitem{GZ16}
P.~Gaborit and G.~Z{\'{e}}mor, ``On the hardness of the decoding and the
  minimum distance problems for rank codes,'' \emph{{IEEE} Trans. Information
  Theory}, vol.~62, no.~12, pp. 7245--7252, 2016.

\bibitem{CS96}
F.~Chabaud and J.~Stern, ``The cryptographic security of the syndrome decoding
  problem for rank distance codes,'' in \emph{Advances in Cryptology -
  ASIACRYPT~1996}, ser. Lecture Notes in Comput. Sci., vol. 1163.\hskip 1em
  plus 0.5em minus 0.4em\relax Kyongju, Korea: Springer, Nov. 1996, pp.
  368--381.

\bibitem{P62}
E.~Prange, ``The use of information sets in decoding cyclic codes,''
  \emph{{IRE} Transactions on Information Theory}, vol.~8, no.~5, pp. 5--9,
  1962.

\bibitem{OJ02}
A.~V. Ourivski and T.~Johansson, ``\BIBforeignlanguage{English}{New technique
  for decoding codes in the rank metric and its cryptography applications},''
  \emph{\BIBforeignlanguage{English}{Problems of Information Transmission}},
  vol.~38, no.~3, pp. 237--246, 2002.

\bibitem{GRS16}
P.~Gaborit, O.~Ruatta, and J.~Schrek, ``On the complexity of the rank syndrome
  decoding problem,'' \emph{{IEEE} Trans. Information Theory}, vol.~62, no.~2,
  pp. 1006--1019, 2016.

\bibitem{AGHT18}
N.~Aragon, P.~Gaborit, A.~Hauteville, and J.~Tillich, ``A new algorithm for
  solving the rank syndrome decoding problem,'' in \emph{2018 {IEEE}
  International Symposium on Information Theory, {ISIT}}.\hskip 1em plus 0.5em
  minus 0.4em\relax {IEEE}, 2018, pp. 2421--2425.

\bibitem{LP06}
F.~L{\'{e}}vy-dit Vehel and L.~Perret, ``Algebraic decoding of codes in rank
  metric,'' in \emph{proceedings of YACC06}, Porquerolles, France, Jun. 2006,
  available on {\tt http://grim.univ-tln.fr/YACC06/abstracts-yacc06.pdf}.

\bibitem{BBCGPSTV20}
M.~Bardet, M.~Bros, D.~Cabarcas, P.~Gaborit, R.~A. Perlner, D.~Smith{-}Tone,
  J.~Tillich, and J.~A. Verbel, ``Improvements of algebraic attacks for solving
  the rank decoding and minrank problems,'' in \emph{Advances in Cryptology -
  {ASIACRYPT}}, ser. Lecture Notes in Computer Science, vol. 12491.\hskip 1em
  plus 0.5em minus 0.4em\relax Springer, 2020, pp. 507--536.

\bibitem{BBBGNRT20}
M.~Bardet, P.~Briaud, M.~Bros, P.~Gaborit, V.~Neiger, O.~Ruatta, and
  J.~Tillich, ``An algebraic attack on rank metric code-based cryptosystems,''
  in \emph{Advances in Cryptology - {EUROCRYPT}}, ser. Lecture Notes in
  Computer Science, A.~Canteaut and Y.~Ishai, Eds., vol. 12107.\hskip 1em plus
  0.5em minus 0.4em\relax Springer, 2020, pp. 64--93.

\bibitem{Puchinger2021efficient}
S.~Puchinger, J.~Renner, A.~Wachter-Zeh, and J.~Zumbr{\"a}ge, ``Efficient
  decoding of gabidulin codes over galois rings,'' in \emph{2021 IEEE
  International Symposium on Information Theory (ISIT)}.\hskip 1em plus 0.5em
  minus 0.4em\relax IEEE, 2021, pp. 25--30.

\bibitem{G95}
K.~Gibson, ``Severely denting the {Gabidulin} version of the {McEliece} public
  key cryptosystem,'' \emph{Des. Codes Cryptogr.}, vol.~6, no.~1, pp. 37--45,
  1995.

\bibitem{G96}
\BIBentryALTinterwordspacing
------, ``The security of the {Gabidulin} public key cryptosystem,'' in
  \emph{Advances in Cryptology - EUROCRYPT '96}, ser. Lecture Notes in Comput.
  Sci., U.~Maurer, Ed., vol. 1070.\hskip 1em plus 0.5em minus 0.4em\relax
  Springer, 1996, pp. 212--223. [Online]. Available:
  \url{http://dx.doi.org/10.1007/3-540-68339-9_19}
\BIBentrySTDinterwordspacing

\bibitem{O05}
R.~Overbeck, ``A new structural attack for {GPT} and variants,'' in
  \emph{Mycrypt}, ser. Lecture Notes in Comput. Sci., vol. 3715, 2005, pp.
  50--63.

\bibitem{O05a}
------, ``Extending {G}ibson's attacks on the {GPT} cryptosystem,'' in
  \emph{WCC 2005}, ser. Lecture Notes in Comput. Sci., O.~Ytrehus, Ed., vol.
  3969.\hskip 1em plus 0.5em minus 0.4em\relax Springer, 2005, pp. 178--188.

\bibitem{O08}
------, ``Structural attacks for public key cryptosystems based on {Gabidulin}
  codes,'' \emph{J. Cryptology}, vol.~21, no.~2, pp. 280--301, 2008.

\bibitem{Gabidulin2001modified}
E.~M. Gabidulin and A.~V. Ourivski, ``Modified gpt pkc with right scrambler,''
  \emph{Electronic Notes in Discrete Mathematics}, vol.~6, pp. 168--177, 2001.

\bibitem{Gabidulin2008attacks}
E.~M. Gabidulin, ``Attacks and counter-attacks on the gpt public key
  cryptosystem,'' \emph{Designs, Codes and Cryptography}, vol.~48, no.~2, pp.
  171--177, 2008.

\bibitem{Rashwan2010smart}
H.~Rashwan, E.~M. Gabidulin, and B.~Honary, ``A smart approach for gpt
  cryptosystem based on rank codes,'' in \emph{2010 IEEE International
  Symposium on Information Theory}.\hskip 1em plus 0.5em minus 0.4em\relax
  IEEE, 2010, pp. 2463--2467.

\bibitem{Loidreau2010designing}
P.~Loidreau, ``Designing a rank metric based mceliece cryptosystem,'' in
  \emph{International Workshop on Post-Quantum Cryptography}.\hskip 1em plus
  0.5em minus 0.4em\relax Springer, 2010, pp. 142--152.

\bibitem{Rashwan2011security}
H.~Rashwan, E.~M. Gabidulin, and B.~Honary, ``Security of the gpt cryptosystem
  and its applications to cryptography,'' \emph{Security and Communication
  Networks}, vol.~4, no.~8, pp. 937--946, 2011.

\bibitem{Otmani2018improved}
A.~Otmani, H.~T. Kalachi, and S.~Ndjeya, ``Improved cryptanalysis of rank
  metric schemes based on gabidulin codes,'' \emph{Designs, Codes and
  Cryptography}, vol.~86, no.~9, pp. 1983--1996, 2018.

\bibitem{Horlemann2018extension}
A.-L. Horlemann-Trautmann, K.~Marshall, and J.~Rosenthal, ``Extension of
  overbeck's attack for gabidulin-based cryptosystems,'' \emph{Designs, Codes
  and Cryptography}, vol.~86, no.~2, pp. 319--340, 2018.

\bibitem{Kalachi2022failure}
H.~T. Kalachi, ``On the failure of the smart approach of the {GPT}
  cryptosystem,'' \emph{Cryptologia}, vol.~46, no.~2, pp. 167--182, 2022.

\bibitem{Kamwa2021generalization}
F.~R. Kamwa~Djomou, H.~Tal{\'e}~Kalachi, and E.~Fouotsa, ``Generalization of
  low rank parity-check (lrpc) codes over the ring of integers modulo a
  positive integer,'' \emph{Arabian Journal of Mathematics}, vol.~10, no.~2,
  pp. 357--366, 2021.

\bibitem{Kamche2021low}
H.~T. Kamche, H.~T. Kalachi, F.~R.~K. Djomou, and E.~Fouotsa, ``Low-rank
  parity-check codes over finite commutative rings and application to
  cryptography,'' \emph{arXiv preprint arXiv:2106.08712}, 2021.

\bibitem{Mcdonald1974finite}
B.~R. McDonald, \emph{Finite rings with identity}.\hskip 1em plus 0.5em minus
  0.4em\relax Marcel Dekker Incorporated, 1974, vol.~28.

\bibitem{Walker1999algebraic}
J.~L. Walker, ``Algebraic geometric codes over rings,'' \emph{Journal of pure
  and applied Algebra}, vol. 144, no.~1, pp. 91--110, 1999.

\bibitem{Greferath2004finite}
M.~Greferath, A.~Nechaev, and R.~Wisbauer, ``Finite quasi-frobenius modules and
  linear codes,'' \emph{Journal of Algebra and its Applications}, vol.~3,
  no.~03, pp. 247--272, 2004.

\bibitem{Dougherty2009mds}
S.~T. Dougherty, J.-L. Kim, and H.~Kulosman, ``{MDS} codes over finite
  principal ideal rings,'' \emph{Designs, Codes and Cryptography}, vol.~50,
  no.~1, p.~77, 2009.

\bibitem{Fan2014matrix}
Y.~Fan, S.~Ling, and H.~Liu, ``Matrix product codes over finite commutative
  {F}robenius rings,'' \emph{Designs, codes and cryptography}, vol.~71, no.~2,
  pp. 201--227, 2014.

\bibitem{Anderson2012rings}
F.~W. Anderson and K.~R. Fuller, \emph{Rings and categories of modules}.\hskip
  1em plus 0.5em minus 0.4em\relax Springer Science \& Business Media, 2012,
  vol.~13.

\bibitem{Honold2000linear}
T.~Honold and I.~Landjev, ``Linear codes over finite chain rings,'' \emph{the
  electronic journal of combinatorics}, vol.~7, pp. R11--R11, 2000.

\bibitem{sagemath2022}
{The SageMath Developers}, \emph{SageMath mathematics software}, 2022,
  http://www.sagemath.org/.

\bibitem{Koetter2008coding}
R.~Koetter and F.~R. Kschischang, ``Coding for errors and erasures in random
  network coding,'' \emph{IEEE Transactions on Information theory}, vol.~54,
  no.~8, pp. 3579--3591, 2008.

\bibitem{Brualdi1977introductory}
R.~A. Brualdi, \emph{Introductory combinatorics}.\hskip 1em plus 0.5em minus
  0.4em\relax Pearson Education India, 1977.

\bibitem{Horn1991topics}
R.~A. Horn and C.~R. Johnson, \emph{Topics in Matrix Analysis}.\hskip 1em plus
  0.5em minus 0.4em\relax Cambridge University Press, 1991.

\bibitem{Bulyovszky2017polynomial}
B.~Bulyovszky and G.~Horv{\'a}th, ``Polynomial functions over finite
  commutative rings,'' \emph{Theoretical Computer Science}, vol. 703, pp.
  76--86, 2017.

\end{thebibliography}

\end{document}